\documentclass[%
 reprint,
 prl,
 amsmath,amssymb,
 aps,
]{revtex4-2}
\usepackage{physics}
\usepackage{tikz}
\usepackage{mathdots}
\usepackage{yhmath}
\usepackage{upgreek}
\usepackage{cancel}
\usepackage{color}
\usepackage{siunitx}
\usepackage{array}
\usepackage{multirow}
\usepackage{amssymb}
\usepackage{url}
\usepackage{hyperref}
\usepackage{gensymb}
\usepackage{tabularx}
\usepackage{extarrows}
\usepackage{booktabs}
\usetikzlibrary{fadings}
\usetikzlibrary{patterns}
\usetikzlibrary{shadows.blur}
\usetikzlibrary{shapes}
\usepackage{graphicx}% Include figure files
\usepackage{dcolumn}% Align table columns on decimal point
\usepackage{bm}% bold math
    \usepackage{amsthm}
\newtheorem{theorem}{Theorem}

\newtheorem{lemma}[theorem]{Lemma}

\newtheorem{definition}{\normalfont Definition}
\newtheorem{convention}{\normalfont Convention}
\def\Mg{$(M, g)$}
\def\mn{\mu \nu }
\def\ls{\mathcal{L}^{\pm }_{l[\sigma ]}}

\begin{document}
\def\Mg{$(M, g)$}
\def\mn{\mu \nu }
\def\ls{\mathcal{L}^{+}_{l[\sigma ]}}
\def\define#1{(definition: {\it #1})}
\preprint{APS/123-QED}

\title{Singularities from Hyperentropic regions using the Quantum Expansion}%

\author{Vaibhav Kalvakota}%
 \email{vaibhavkalvakota@icloud.com}
\affiliation{%
Turito Institute, 500081, Hyderabad, India}%
 
\date{\today}% It is always \today, today,

\begin{abstract}
A recent paper \cite{Bousso:2022cun} put forward a theorem showing that hyperentropic surface would result in incomplete null generators for a null hypersurface emanating from the surface provided it satisfies the null curvature condition and the spacetime is globally hyperbolic. In this paper, we will put forward a version of this theorem using the quantum expansion in place of the classical expansion, and we will discuss the quantum focusing conjecture in this regard.
\end{abstract}

\maketitle

We propose a quantum version of the theorem \cite{Bousso:2022cun}, which can be stated as:
\begin{theorem}\label{BoussoArvin}
Let $(M, g)$ be a globally hyperbolic spacetime with a codimension $2$ surface $I$ with a compact boundary $\partial I$. If $S(I)>\frac{A(\partial I)}{4G\hbar }$ and if the expansion is negative along the future ingoing null congruence that is orthogonal to $\partial I$ under the satisfaction of the null curvature condition, then at least one null generator is incomplete.
\end{theorem}
The Bousso bound \cite{Bousso:1999xyz, Flanagan:1999jp} is the statement that for a surface with at least one of the four principle null surface-orthogonal congruences having a negative expansion, the entropy bound on the surface is given by the area of the boundary of the surface (when there exist more than one light-sheet, we simply add in those factors -- the maximal form of the inequality is with $S\leq A$):
\begin{equation}\label{eq:Boussobound}
    S(\sigma )\leq \frac{A(\partial \sigma )}{4G\hbar }.
\end{equation}
This is interpreted as the entropy that is bound through the light-sheet, and since the light-sheet is compact (terminating either at a focal point or a caustic) \cite{Bousso:1999xyz}, the boundary of the light-sheet is the boundary of the surface itself. 
In the paper \cite{Bousso:2022cun}, Bousso and Moghaddam introduced a classical singularity theorem that states that if a surface has an entropy greater than what the Bousso bound allows, i.e. if 
\begin{equation}
    S(I)>\frac{A(\partial I)}{4G\hbar },
\end{equation}
then under theorem \ref{BoussoArvin}, the region (called a \textit{hyperentropic} region) must necessarily have at least one incomplete null generator provided the surface is at least future marginally trapped. This is an interesting result particularly since this can be interpreted as a singularity arising because of too much information in a given surface. A possible converse of this theorem will be discussed in an upcoming paper, which would hint at the existence of a hyperentropic region in a spacetime with a null geodesic incompleteness. As we will state and prove in this paper, the result in theorem \ref{BoussoArvin} has implications from the generalized second law, which can be stated as the following variation condition
\begin{equation}
    \delta S_{gen}(\sigma , \Sigma )>0,
\end{equation}
where $S_{gen}$ is the generalized entropy for a codimension $2$ surface $\sigma $ in a Cauchy slice $\Sigma $. Throughout this paper, by a surface $I$ we necessarily imply a codimension $2$ surface that has at least one congruence with a negative expansion, while by $\sigma $ we will imply a codimension $2$ surface without such restriction. 

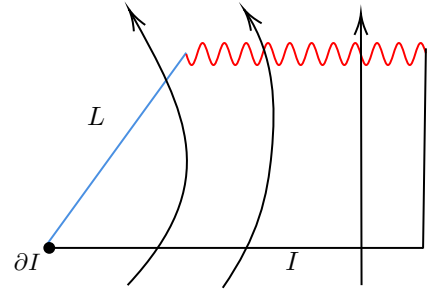
\begin{center}

\tikzset{every picture/.style={line width=0.75pt}} %set default line width to 0.75pt        

\begin{tikzpicture}[x=0.75pt,y=0.75pt,yscale=-1,xscale=1]

%Straight Lines [id:da33590700359248205] 
\draw    (247.87,191.09) -- (438.87,191.1) ;
%Straight Lines [id:da4191707241936149] 
\draw [color={rgb, 255:red, 74; green, 144; blue, 226 }  ,draw opacity=1 ]   (247.87,191.09) -- (319.42,92.58) ;
%Shape: Wave [id:dp4969120503455575] 
\draw  [color={rgb, 255:red, 255; green, 0; blue, 0 }  ,draw opacity=1 ] (319.42,93.23) .. controls (320.31,96.08) and (321.17,98.79) .. (322.17,98.79) .. controls (323.16,98.79) and (324.02,96.08) .. (324.92,93.23) .. controls (325.81,90.38) and (326.67,87.67) .. (327.67,87.67) .. controls (328.66,87.67) and (329.52,90.38) .. (330.42,93.23) .. controls (331.31,96.08) and (332.17,98.79) .. (333.17,98.79) .. controls (334.16,98.79) and (335.02,96.08) .. (335.92,93.23) .. controls (336.81,90.38) and (337.67,87.67) .. (338.67,87.67) .. controls (339.66,87.67) and (340.52,90.38) .. (341.42,93.23) .. controls (342.31,96.08) and (343.17,98.79) .. (344.17,98.79) .. controls (345.16,98.79) and (346.02,96.08) .. (346.92,93.23) .. controls (347.81,90.38) and (348.67,87.67) .. (349.67,87.67) .. controls (350.66,87.67) and (351.52,90.38) .. (352.42,93.23) .. controls (353.31,96.08) and (354.17,98.79) .. (355.17,98.79) .. controls (356.16,98.79) and (357.02,96.08) .. (357.92,93.23) .. controls (358.81,90.38) and (359.67,87.67) .. (360.67,87.67) .. controls (361.66,87.67) and (362.52,90.38) .. (363.42,93.23) .. controls (364.31,96.08) and (365.17,98.79) .. (366.17,98.79) .. controls (367.16,98.79) and (368.02,96.08) .. (368.92,93.23) .. controls (369.81,90.38) and (370.67,87.67) .. (371.67,87.67) .. controls (372.66,87.67) and (373.52,90.38) .. (374.42,93.23) .. controls (375.31,96.08) and (376.17,98.79) .. (377.17,98.79) .. controls (378.16,98.79) and (379.02,96.08) .. (379.92,93.23) .. controls (380.81,90.38) and (381.67,87.67) .. (382.67,87.67) .. controls (383.66,87.67) and (384.52,90.38) .. (385.42,93.23) .. controls (386.31,96.08) and (387.17,98.79) .. (388.17,98.79) .. controls (389.16,98.79) and (390.02,96.08) .. (390.92,93.23) .. controls (391.81,90.38) and (392.67,87.67) .. (393.67,87.67) .. controls (394.66,87.67) and (395.52,90.38) .. (396.42,93.23) .. controls (397.31,96.08) and (398.17,98.79) .. (399.17,98.79) .. controls (400.16,98.79) and (401.02,96.08) .. (401.92,93.23) .. controls (402.81,90.38) and (403.67,87.67) .. (404.67,87.67) .. controls (405.66,87.67) and (406.52,90.38) .. (407.42,93.23) .. controls (408.31,96.08) and (409.17,98.79) .. (410.17,98.79) .. controls (411.16,98.79) and (412.02,96.08) .. (412.92,93.23) .. controls (413.81,90.38) and (414.67,87.67) .. (415.67,87.67) .. controls (416.66,87.67) and (417.52,90.38) .. (418.42,93.23) .. controls (419.31,96.08) and (420.17,98.79) .. (421.17,98.79) .. controls (422.16,98.79) and (423.02,96.08) .. (423.92,93.23) .. controls (424.81,90.38) and (425.67,87.67) .. (426.67,87.67) .. controls (427.66,87.67) and (428.52,90.38) .. (429.42,93.23) .. controls (430.31,96.08) and (431.17,98.79) .. (432.17,98.79) .. controls (433.16,98.79) and (434.02,96.08) .. (434.92,93.23) .. controls (435.81,90.38) and (436.67,87.67) .. (437.67,87.67) .. controls (438.62,87.67) and (439.44,90.13) .. (440.29,92.83) ;
%Straight Lines [id:da7932513573057052] 
\draw    (440.29,90.37) -- (438.87,191.42) ;
%Curve Lines [id:da6354632465154313] 
\draw    (289.84,209.89) .. controls (338.03,158.16) and (320.1,123.36) .. (291.39,72.19) ;
\draw [shift={(290.52,70.63)}, rotate = 60.68] [color={rgb, 255:red, 0; green, 0; blue, 0 }  ][line width=0.75]    (10.93,-3.29) .. controls (6.95,-1.4) and (3.31,-0.3) .. (0,0) .. controls (3.31,0.3) and (6.95,1.4) .. (10.93,3.29)   ;
%Curve Lines [id:da44667751587864246] 
\draw    (337.84,211.22) .. controls (355.9,184.97) and (360.38,165.65) .. (362.52,142.63) .. controls (364.62,120.08) and (364.29,94.36) .. (349.45,72) ;
\draw [shift={(348.52,70.63)}, rotate = 55.2] [color={rgb, 255:red, 0; green, 0; blue, 0 }  ][line width=0.75]    (10.93,-3.29) .. controls (6.95,-1.4) and (3.31,-0.3) .. (0,0) .. controls (3.31,0.3) and (6.95,1.4) .. (10.93,3.29)   ;
%Straight Lines [id:da47959056091897345] 
\draw    (407.84,209.89) -- (407.52,74.63) ;
\draw [shift={(407.52,72.63)}, rotate = 89.87] [color={rgb, 255:red, 0; green, 0; blue, 0 }  ][line width=0.75]    (10.93,-3.29) .. controls (6.95,-1.4) and (3.31,-0.3) .. (0,0) .. controls (3.31,0.3) and (6.95,1.4) .. (10.93,3.29)   ;
%Shape: Circle [id:dp4717146654850801] 
\draw  [fill={rgb, 255:red, 0; green, 0; blue, 0 }  ,fill opacity=1 ] (247.87,191.09) .. controls (247.87,189.73) and (248.98,188.62) .. (250.34,188.62) .. controls (251.71,188.62) and (252.82,189.73) .. (252.82,191.09) .. controls (252.82,192.46) and (251.71,193.56) .. (250.34,193.56) .. controls (248.98,193.56) and (247.87,192.46) .. (247.87,191.09) -- cycle ;

% Text Node
\draw (231,192.4) node [anchor=north west][inner sep=0.75pt]  [font=\small]  {$\partial I$};
% Text Node
\draw (368,191.4) node [anchor=north west][inner sep=0.75pt]    {$I$};
% Text Node
\draw (268,118.4) node [anchor=north west][inner sep=0.75pt]    {$L$};

\end{tikzpicture}
\end{center}

The theorem we wish to prove in relation to theorem \ref{BoussoArvin} is as follows, where $\Uptheta $ is the quantum expansion:
\begin{theorem}\label{qba}
Let a surface $I$ have a boundary $\partial I$ that is a compact codimension $2$ submanifold with splitting, and let $\Uptheta ^{+}_{l}<0$. If the surface is hyperentropic, there exists at least one null incomplete generator.
\end{theorem}

We will now discuss some definitions and propositions that will help us in reviewing this theorem and discuss the implications of the proof \cite{Bousso:2022cun}.

\begin{convention}
A spacetime $(M, g)$ is said to be globally hyperbolic if $J^{+}(p)\cap J^{-}(q)$ is compact for $p, q\in M$, and if there exists a Cauchy surface $\Sigma $ in $(M, g)$. We will assume throughout this paper that $(M ,g)$ is globally hyperbolic. To indicate the boundary of a surface $I$ in $\Sigma $, we use $\partial I$, and to indicate the boundary of a surface in $(M, g)$, we will use $\Dot{I}$ \cite{Hawking:1973uf}. 
\end{convention}
\begin{convention}
By a boundary $\partial \sigma $ we mean a compact Cauchy splitting codimension $2$ submanifold for a surface $\sigma $ that defines an interior and exterior for the Cauchy slice $\Sigma $ and $\sigma -\partial \sigma \neq \emptyset $. For $\sigma =I$, we would mean a surface for which at least the future ingoing congruence of the principle null $I-$orthogonal directions has a negative expansion. 
\end{convention}
\begin{definition}
By the future (resp. past) domain of dependence of a set $A\subset M$, we mean the set of points $p\in M$ such that every past (resp. future) inextendible timelike curve $\gamma $ passing through $p$ also necessarily intersects $A$. The domain of dependence $D(A)=D^{+}(A)\cup D^{-}(A)$.
\end{definition}
\begin{definition}
By a light-sheet $L$ for a surface $I$, we mean a null hypersurface formed by all null geodesics originating orthogonal to $I$ such that $d^{\pm }\theta ^{\pm }_{L}\leq 0$ and for which $I\subset \partial L$.
\end{definition}
\begin{definition}
(Null curvature condition): If $(M, g)$ is such that for all null vectors $k$ we have
\begin{equation}
    R_{\mu \nu }k^{\mu }k^{\nu }\geq 0,
\end{equation}
then we say that $M$ satisfies the null curvature condition.
\end{definition}
Note that, in itself the null curvature condition does not necessarily imply the Bousso bound -- rather, the role of this condition is to ensure that given a non-positive expansion $\theta _{0}$, under the Raychaudari equation
\begin{equation}
    \frac{d\theta }{d\lambda }=-\frac{1}{2}\theta ^{2}-8\pi GT_{\mu \nu }k^{\mu }k^{\nu }-\sigma _{\mu \nu }\sigma ^{\mu \nu }+\omega _{\mu \nu }\omega ^{\mu \nu },
\end{equation}
(where the vorticity tensor vanishes since we are considering orthogonal null congruences) the variation of the expansion will be non-positive throughout the light-sheet. 
\begin{definition}
The Bousso bound states that for a surface $I$ with a light-sheet $L$ such that $L$ does not terminate before a focal point (i.e. there are no dense regions contributing to a closing or terminating light-sheet before a point $b$ on $L$ as defined below), if $\partial I=\partial L$, then the entropy bound by the light-sheet in terms of the boundary $\partial I$ is given by \eqref{eq:Boussobound}.
\end{definition}
\begin{lemma}
If a point $p\in M$ is on the future boundary of a surface $I$, then the following conditions must hold \cite{Akers:2017nrr}:
\begin{enumerate}
    \item there are no intersecting null geodesics before $p$,
    \item $p$ lies on a null orthogonal geodesic $\gamma $ emanating from $\partial I$, and 
    \item there are no conjugate points on $\gamma $ before $p$.
\end{enumerate}
\end{lemma}
In this paper, we will assume that there are no caustics interfering with the light-sheet. Therefore, the light-sheet would normally close off if the region were non-hyperentropic. We now turn our attention towards some lemmas that will prove theorem \ref{BoussoArvin}, which will further help us in stating and proving theorem \ref{qba}.
\begin{lemma}
Let $X$ denote the complement of $I$ in $\Sigma $. Then, $D^{+}(I)=D^{+}(\Sigma )-\mathcal{I}^{+}(X)-X$. Further, $D(I)=M-\mathcal{I}(X)$.
\end{lemma}
\begin{proof}
Let a point $p\in M$ be $p\in D^{+}(\Sigma )-\mathcal{I}^{+}(X)-X$. Then, a past inextendible timelike curve passing through $p$ would necessarily be in $D^{+}(\Sigma )$, while the subtracted $\mathcal{I}^{+}(X)-X$ would mean that such a curve would not intersect $X$. Due to this, such curves would intersect $I$, and therefore $p\in D^{+}(I)$. Similarly, the same can be said for the past domain of dependence. Due to this, we can state that the domain of dependence of $I$ is $M-\mathcal{I}(X)$.
\end{proof}
\begin{lemma}
$X$ has the property that $\Dot{\mathcal{I}}^{+}(X)-X=\Dot{\mathcal{I}}^{+}(\partial X)-\mathcal{I}^{+}(X)-X$.
\end{lemma}
The proof of this lemma can be read from \cite{Bousso:2022cun}, which we will not reproduce here. We will instead use this to state the following lemma:
\begin{lemma}
$\Dot{\mathcal{I}}^{+}(X)-X$ is the future outgoing null congruence from $I$ and terminates at either a caustic or at a point where neighbouring null geodesics intersect.
\end{lemma}
\begin{convention}
We will adopt the following convention: future (resp. past) outgoing null orthogonal congruence will be labelled as $K^{\pm }$, while the future (resp. past) ingoing null orthogonal congruence will be labelled as $L^{\pm }$. 
\end{convention}
\begin{proof}
This follows from \cite{Akers:2017nrr}, where naturally the null geodesics that compose the future outgoing congruence do not enter $\mathcal{I}^{+}(X)$ (i.e. exist the future boundary of $K^{+}(I)$) unless we encounter a caustic or self-intersection. Due to this and the previous lemma, we can state that the set $\Dot{\mathcal{I}}^{+}(X)-X$ is generated by the future null orthogonal geodesics originating from $\partial I$.
\end{proof}
The proof of theorem \ref{BoussoArvin} is as follows:
\begin{proof}
The light-sheet $L^{+}=\Dot{\mathcal{I}}^{+}(X)-X$ is a surface such that the boundary $\partial L=\partial I$, and from lemma 6 this would be a null hypersurface. Further, since the expansion condition $\theta <0$ holds on $\partial I$ along $L$, due to the null curvature condition it must be negative throughout the hypersurface (before encountering a caustic), implying $L$ is a light-sheet. If we assume that all generators of $L$ are necessarily complete, it must also be true that the light-sheet is compact, since the light-sheet generators must leave $L$ after a fixed rescaled affine length, implying a closed $L$. Now, since $\partial L=\partial I$, the domain of dependence $D(L)=D(I)$ (from proposition 10 in \cite{Bousso:2022cun})  and therefore the Bousso bound would be contradicted, since by assumption the region $I$ is hyperentropic, and therefore the light-sheet must not close to preserve the Bousso bound, contradicting our assumption that all the null generators of $L$ are necessarily complete.
\end{proof}

It should be possible to extend this result into the quantum limit, where the expansion takes a slightly different meaning. In the classical limit, one would expect the expansion to be the measure of the rate of change of area via deformations, while in the quantum regimes, one would expect this to be a measure of the rate of change of \textit{generalized entropy} \cite{Bekenstein:1974ax} via deformations:
\begin{equation}\label{eq:qexp}
    \Uptheta =\lim _{\mathcal{A}\to 0}\; \frac{4G\hbar }{\mathcal{A}}\frac{dS_{gen}(\sigma , \Sigma )}{d\lambda }|_{x\in \sigma }.
\end{equation}
In general, surfaces are defined as:

\begin{center}
\begin{tabular}{ |c|c|c| } 
 \hline
  & $\theta ^{+}_{k}$ & $\theta ^{+}_{l}$ \\ 
 Q Trapped & $-$ & $-$ \\ 
 Q marginally trapped & $0$ & $-$ \\ 
 Q Untrapped & $+$ & $-$ \\
 Q Extremal & $0$ & $0$ \\
 \hline
\end{tabular}
\label{table:expsign}
\end{center}

Before we prove the theorem \ref{qba}, we will first look at Aron Wall's result \cite{Wall:2010jtc} from generalized entropy, which goes analogously to that of Penrose's theorem. In fact, Penrose's statement can be written into stating that if one has a non-compact region with a compact boundary with the expansion along at least one future null surface orthogonal congruence being negative, due to global hyperbolicity of $(M, g)$ a non-compact slice cannot evolve into a compact slice, which is a topological constraint using which we state that there has to exist some null geodesic incompleteness. That is, if we have a compact splitting boundary with the surface compact and if the expansion condition holds, then we can state that the light-sheet would form $L=\partial ^{+}D(I)$, and then at least one generator of this must be incomplete. This is essentially the backing to the theorems discussed previously -- we can consider $L$ to be homeomorphically identified to $I$, and therefore if $L$ is compact, $I$ must be compact as well. If all the congruences have a negative expansion, we call the surface $I$ a trapped surface. Similarly, one defines a quantum trapped surface as follows: let a compact $I$ split the Cauchy slice $\Sigma $ into an interior and an exterior region. Then, the future boundary of the domain of dependence of the exterior $\text{Ext}(I)$ is defined as a null surface $K^{+}$ constructed from future outgoing light rays from $I$. If there is a way to evolve the slice so that the fine-grained generalized entropy is decreasing along $L$, then we say that $I$ has a negative quantum expansion throughout $L$ under the quantum focusing conjecture. Wall's semiclassical theorem \cite{Wall:2010jtc} can be stated as the following for the future ingoing null congruence with a negative expansion:
\begin{theorem}\label{wall1}
Let $I$ be a non-compact surface with a compact boundary $\partial I$ as considered previously. Then, if the quantum expansion $\Uptheta (l)<0$, then there exists at least one incomplete null generator.
\end{theorem}
\begin{proof}
This is similar to the classical Penrose's theorem discussed above. Since the generalized entropy is decreasing, the compactness of $L$ in terms of the generators can be found.  Defining a timelike vector field in $(M ,g)$, every curve intersecting $L$ would also intersect $I$ under a continuous $1-1$ mapping $\varphi :L\to I$ (this can be replaced by a non-compact Cauchy surface to get Penrose's theorem). Since the image $\varphi (L(I))$ would also be compact as $L$ is compact, this would not cover all of $I$, since this is non-compact. Due to this, there must exist at least one incomplete null generator. 
\end{proof}
We now prove theorem \ref{qba} with respect to the quantum focusing condition.
\begin{proof}
Since the generalized entropy on the null hypersurface $L$ is decreasing, the generator segments on $L$ must terminate at some finite affine value. Here on, the proof will take up a similar structure as of the original theorem proof, the similarity being that in the classical version the classical expansion is negative, while in this case the quantum expansion being negative implies a decreasing generalized entropy. Note that this would imply a violation of the generalized second law when the surface in question is a horizon, in which case we would have to reconsider the entire scheme. However, we are considering this surface to simply have the future ingoing null orthogonal congruence expansion to be negative, and we do not make statements about the outgoing congruence. 

Similarly to the classical expansion, the quantum expansion being negative implies that the generators must leave $L$ after some finite affine length, which would again imply a closed $L$. Since the boundaries of $L$ as well as that of $I$ are the same, the domains of dependence become $D(L)=D(I)$ \cite{Bousso:2022cun}, and due to this the entropies must satisfy $S(L)=S(I)$. Since $S(I)$ was assumed to be hyperentropic, this would imply a violation of the Bousso bound (in a sense we will discuss below), and therefore the light-sheet cannot close off. 
\end{proof}

From this result one could argue that the entropy contained in a surface with a compact boundary has to be in satisfaction with the Bousso bound, and the "overflow" of information in a region of spacetime affects the geodesics passing through the surface. The generalized entropy perspective reduces to the classical expansion, stating that the null hypersurface $L$ is compact. 

The implications of this theorem are straightforward -- let $g_{0}$ be a point on a null generator of the null hypersurface $L$ in consideration. Then, if the semiclassical approximation is valid around $g_{0}$, and if the surface $I$ is hyperentropic, then by evolving $I$ in time so that the generalized entropy on the null hypersurface $L$ is decreasing, the generalized second law would require that at least one null generator of $L$ is incomplete. If this null hypersurface were a horizon, the GSL would be violated \cite{Wall:2010jtc}. This can be viewed in the cosmological perspective by identifying a compact region such that the relative scaling between $A(I)$ and $S(I)$ is such that the region becomes hyperentropic, such as in the case of an flat expanding cosmology (refer to figure below). If $I$ were non-compact, this would imply a null geodesic incompleteness directly from theorem \ref{wall1}, and the past light-sheet would run into a big bang singularity. In the semiclassical result theorem \ref{qba} we get a similar result, for a case when the semiclassical approximation holds for a generator point $g_{0}$. 

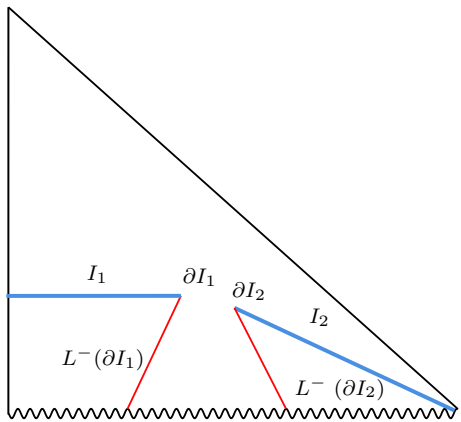
\begin{center}

\tikzset{every picture/.style={line width=0.70pt}} %set default line width to 0.75pt        

\begin{tikzpicture}[x=0.70pt,y=0.70pt,yscale=-0.85,xscale=0.85]

%Shape: Wave [id:dp6479819359763381] 
\draw   (163.73,284.93) .. controls (164.36,286.47) and (164.96,287.93) .. (165.66,287.93) .. controls (166.36,287.93) and (166.97,286.47) .. (167.6,284.93) .. controls (168.23,283.39) and (168.84,281.93) .. (169.54,281.93) .. controls (170.24,281.93) and (170.84,283.39) .. (171.48,284.93) .. controls (172.11,286.47) and (172.71,287.93) .. (173.41,287.93) .. controls (174.12,287.93) and (174.72,286.47) .. (175.35,284.93) .. controls (175.98,283.39) and (176.59,281.93) .. (177.29,281.93) .. controls (177.99,281.93) and (178.6,283.39) .. (179.23,284.93) .. controls (179.86,286.47) and (180.46,287.93) .. (181.17,287.93) .. controls (181.87,287.93) and (182.47,286.47) .. (183.1,284.93) .. controls (183.74,283.39) and (184.34,281.93) .. (185.04,281.93) .. controls (185.74,281.93) and (186.35,283.39) .. (186.98,284.93) .. controls (187.61,286.47) and (188.22,287.93) .. (188.92,287.93) .. controls (189.62,287.93) and (190.22,286.47) .. (190.85,284.93) .. controls (191.49,283.39) and (192.09,281.93) .. (192.79,281.93) .. controls (193.49,281.93) and (194.1,283.39) .. (194.73,284.93) .. controls (195.36,286.47) and (195.97,287.93) .. (196.67,287.93) .. controls (197.37,287.93) and (197.97,286.47) .. (198.61,284.93) .. controls (199.24,283.39) and (199.84,281.93) .. (200.54,281.93) .. controls (201.24,281.93) and (201.85,283.39) .. (202.48,284.93) .. controls (203.11,286.47) and (203.72,287.93) .. (204.42,287.93) .. controls (205.12,287.93) and (205.73,286.47) .. (206.36,284.93) .. controls (206.99,283.39) and (207.59,281.93) .. (208.29,281.93) .. controls (209,281.93) and (209.6,283.39) .. (210.23,284.93) .. controls (210.86,286.47) and (211.47,287.93) .. (212.17,287.93) .. controls (212.87,287.93) and (213.48,286.47) .. (214.11,284.93) .. controls (214.74,283.39) and (215.34,281.93) .. (216.05,281.93) .. controls (216.75,281.93) and (217.35,283.39) .. (217.98,284.93) .. controls (218.62,286.47) and (219.22,287.93) .. (219.92,287.93) .. controls (220.62,287.93) and (221.23,286.47) .. (221.86,284.93) .. controls (222.49,283.39) and (223.1,281.93) .. (223.8,281.93) .. controls (224.5,281.93) and (225.1,283.39) .. (225.74,284.93) .. controls (226.37,286.47) and (226.97,287.93) .. (227.67,287.93) .. controls (228.37,287.93) and (228.98,286.47) .. (229.61,284.93) .. controls (230.24,283.39) and (230.85,281.93) .. (231.55,281.93) .. controls (232.25,281.93) and (232.85,283.39) .. (233.49,284.93) .. controls (234.12,286.47) and (234.72,287.93) .. (235.42,287.93) .. controls (236.13,287.93) and (236.73,286.47) .. (237.36,284.93) .. controls (237.99,283.39) and (238.6,281.93) .. (239.3,281.93) .. controls (240,281.93) and (240.61,283.39) .. (241.24,284.93) .. controls (241.87,286.47) and (242.47,287.93) .. (243.18,287.93) .. controls (243.88,287.93) and (244.48,286.47) .. (245.11,284.93) .. controls (245.75,283.39) and (246.35,281.93) .. (247.05,281.93) .. controls (247.75,281.93) and (248.36,283.39) .. (248.99,284.93) .. controls (249.62,286.47) and (250.23,287.93) .. (250.93,287.93) .. controls (251.63,287.93) and (252.23,286.47) .. (252.86,284.93) .. controls (253.5,283.39) and (254.1,281.93) .. (254.8,281.93) .. controls (255.5,281.93) and (256.11,283.39) .. (256.74,284.93) .. controls (257.37,286.47) and (257.98,287.93) .. (258.68,287.93) .. controls (259.38,287.93) and (259.98,286.47) .. (260.62,284.93) .. controls (261.25,283.39) and (261.85,281.93) .. (262.55,281.93) .. controls (263.26,281.93) and (263.86,283.39) .. (264.49,284.93) .. controls (265.12,286.47) and (265.73,287.93) .. (266.43,287.93) .. controls (267.13,287.93) and (267.74,286.47) .. (268.37,284.93) .. controls (269,283.39) and (269.6,281.93) .. (270.31,281.93) .. controls (271.01,281.93) and (271.61,283.39) .. (272.24,284.93) .. controls (272.88,286.47) and (273.48,287.93) .. (274.18,287.93) .. controls (274.88,287.93) and (275.49,286.47) .. (276.12,284.93) .. controls (276.75,283.39) and (277.36,281.93) .. (278.06,281.93) .. controls (278.76,281.93) and (279.36,283.39) .. (279.99,284.93) .. controls (280.63,286.47) and (281.23,287.93) .. (281.93,287.93) .. controls (282.63,287.93) and (283.24,286.47) .. (283.87,284.93) .. controls (284.5,283.39) and (285.11,281.93) .. (285.81,281.93) .. controls (286.51,281.93) and (287.11,283.39) .. (287.75,284.93) .. controls (288.38,286.47) and (288.98,287.93) .. (289.68,287.93) .. controls (290.38,287.93) and (290.99,286.47) .. (291.62,284.93) .. controls (292.25,283.39) and (292.86,281.93) .. (293.56,281.93) .. controls (294.26,281.93) and (294.86,283.39) .. (295.5,284.93) .. controls (296.13,286.47) and (296.73,287.93) .. (297.43,287.93) .. controls (298.14,287.93) and (298.74,286.47) .. (299.37,284.93) .. controls (300,283.39) and (300.61,281.93) .. (301.31,281.93) .. controls (302.01,281.93) and (302.62,283.39) .. (303.25,284.93) .. controls (303.88,286.47) and (304.48,287.93) .. (305.19,287.93) .. controls (305.89,287.93) and (306.49,286.47) .. (307.12,284.93) .. controls (307.76,283.39) and (308.36,281.93) .. (309.06,281.93) .. controls (309.76,281.93) and (310.37,283.39) .. (311,284.93) .. controls (311.63,286.47) and (312.24,287.93) .. (312.94,287.93) .. controls (313.64,287.93) and (314.24,286.47) .. (314.88,284.93) .. controls (315.51,283.39) and (316.11,281.93) .. (316.81,281.93) .. controls (317.51,281.93) and (318.12,283.39) .. (318.75,284.93) .. controls (319.38,286.47) and (319.99,287.93) .. (320.69,287.93) .. controls (321.39,287.93) and (321.99,286.47) .. (322.63,284.93) .. controls (323.26,283.39) and (323.86,281.93) .. (324.56,281.93) .. controls (325.27,281.93) and (325.87,283.39) .. (326.5,284.93) .. controls (327.13,286.47) and (327.74,287.93) .. (328.44,287.93) .. controls (329.14,287.93) and (329.75,286.47) .. (330.38,284.93) .. controls (331.01,283.39) and (331.61,281.93) .. (332.32,281.93) .. controls (333.02,281.93) and (333.62,283.39) .. (334.25,284.93) .. controls (334.89,286.47) and (335.49,287.93) .. (336.19,287.93) .. controls (336.89,287.93) and (337.5,286.47) .. (338.13,284.93) .. controls (338.76,283.39) and (339.37,281.93) .. (340.07,281.93) .. controls (340.77,281.93) and (341.37,283.39) .. (342,284.93) .. controls (342.64,286.47) and (343.24,287.93) .. (343.94,287.93) .. controls (344.64,287.93) and (345.25,286.47) .. (345.88,284.93) .. controls (346.51,283.39) and (347.12,281.93) .. (347.82,281.93) .. controls (348.52,281.93) and (349.12,283.39) .. (349.76,284.93) .. controls (350.39,286.47) and (350.99,287.93) .. (351.69,287.93) .. controls (352.4,287.93) and (353,286.47) .. (353.63,284.93) .. controls (354.26,283.39) and (354.87,281.93) .. (355.57,281.93) .. controls (356.27,281.93) and (356.88,283.39) .. (357.51,284.93) .. controls (358.14,286.47) and (358.74,287.93) .. (359.45,287.93) .. controls (360.15,287.93) and (360.75,286.47) .. (361.38,284.93) .. controls (362.01,283.39) and (362.62,281.93) .. (363.32,281.93) .. controls (364.02,281.93) and (364.63,283.39) .. (365.26,284.93) .. controls (365.89,286.47) and (366.5,287.93) .. (367.2,287.93) .. controls (367.9,287.93) and (368.5,286.47) .. (369.13,284.93) .. controls (369.77,283.39) and (370.37,281.93) .. (371.07,281.93) .. controls (371.77,281.93) and (372.38,283.39) .. (373.01,284.93) .. controls (373.64,286.47) and (374.25,287.93) .. (374.95,287.93) .. controls (375.65,287.93) and (376.25,286.47) .. (376.89,284.93) .. controls (377.52,283.39) and (378.12,281.93) .. (378.82,281.93) .. controls (379.52,281.93) and (380.13,283.39) .. (380.76,284.93) .. controls (381.39,286.47) and (382,287.93) .. (382.7,287.93) .. controls (383.4,287.93) and (384,286.47) .. (384.64,284.93) .. controls (385.27,283.39) and (385.87,281.93) .. (386.57,281.93) .. controls (387.28,281.93) and (387.88,283.39) .. (388.51,284.93) .. controls (389.14,286.47) and (389.75,287.93) .. (390.45,287.93) .. controls (391.15,287.93) and (391.76,286.47) .. (392.39,284.93) .. controls (393.02,283.39) and (393.62,281.93) .. (394.33,281.93) .. controls (395.03,281.93) and (395.63,283.39) .. (396.26,284.93) .. controls (396.9,286.47) and (397.5,287.93) .. (398.2,287.93) .. controls (398.9,287.93) and (399.51,286.47) .. (400.14,284.93) .. controls (400.77,283.39) and (401.38,281.93) .. (402.08,281.93) .. controls (402.78,281.93) and (403.38,283.39) .. (404.02,284.93) .. controls (404.65,286.47) and (405.25,287.93) .. (405.95,287.93) .. controls (406.65,287.93) and (407.26,286.47) .. (407.89,284.93) .. controls (408.52,283.39) and (409.13,281.93) .. (409.83,281.93) .. controls (410.53,281.93) and (411.13,283.39) .. (411.77,284.93) .. controls (412.4,286.47) and (413,287.93) .. (413.7,287.93) .. controls (414.41,287.93) and (415.01,286.47) .. (415.64,284.93) .. controls (416.27,283.39) and (416.88,281.93) .. (417.58,281.93) .. controls (418.28,281.93) and (418.89,283.39) .. (419.52,284.93) .. controls (420.15,286.47) and (420.75,287.93) .. (421.46,287.93) .. controls (422.16,287.93) and (422.76,286.47) .. (423.39,284.93) .. controls (424.03,283.39) and (424.63,281.93) .. (425.33,281.93) .. controls (426.03,281.93) and (426.64,283.39) .. (427.27,284.93) .. controls (427.9,286.47) and (428.51,287.93) .. (429.21,287.93) .. controls (429.91,287.93) and (430.51,286.47) .. (431.14,284.93) .. controls (431.78,283.39) and (432.38,281.93) .. (433.08,281.93) .. controls (433.78,281.93) and (434.39,283.39) .. (435.02,284.93) .. controls (435.65,286.47) and (436.26,287.93) .. (436.96,287.93) .. controls (437.66,287.93) and (438.26,286.47) .. (438.9,284.93) .. controls (439.53,283.39) and (440.13,281.93) .. (440.83,281.93) .. controls (441.54,281.93) and (442.14,283.39) .. (442.77,284.93) .. controls (443.4,286.47) and (444.01,287.93) .. (444.71,287.93) .. controls (445.41,287.93) and (446.02,286.47) .. (446.65,284.93) .. controls (447.28,283.39) and (447.88,281.93) .. (448.58,281.93) .. controls (448.81,281.93) and (449.02,282.07) .. (449.23,282.33) ;
%Straight Lines [id:da9396072764312473] 
\draw [line width=0.75]    (163.5,26.43) -- (163.48,285.52) ;
%Straight Lines [id:da06439431589304312] 
\draw    (163.5,26.43) -- (449.12,281.92) ;
%Straight Lines [id:da8644598969886306] 
\draw [color={rgb, 255:red, 74; green, 144; blue, 226 }  ,draw opacity=1 ][line width=1.5]    (307.2,217.7) -- (447.05,283.01) ;
%Straight Lines [id:da19337200844641989] 
\draw [color={rgb, 255:red, 255; green, 5; blue, 5 }  ,draw opacity=1 ]   (340.01,281.93) -- (307.2,217.7) ;
%Straight Lines [id:da28640870206297175] 
\draw [color={rgb, 255:red, 255; green, 5; blue, 5 }  ,draw opacity=1 ]   (239.18,282.05) -- (273.18,210.05) ;
%Straight Lines [id:da14934317837486621] 
\draw [color={rgb, 255:red, 74; green, 144; blue, 226 }  ,draw opacity=1 ][line width=1.5]    (162.2,210.1) -- (273.18,210.05) ;

% Text Node
\draw (352.75,216.47) node [anchor=north west][inner sep=0.75pt]  [font=\footnotesize]  {$I_{2}$};
% Text Node
\draw (303.92,199.54) node [anchor=north west][inner sep=0.75pt]  [font=\footnotesize]  {$\partial I_{2}$};
% Text Node
\draw (343.73,259.83) node [anchor=north west][inner sep=0.75pt]  [font=\footnotesize]  {$L^{-\ }( \partial I_{2})$};
% Text Node
\draw (272.59,190.54) node [anchor=north west][inner sep=0.75pt]  [font=\footnotesize]  {$\partial I_{1}$};
% Text Node
\draw (211.5,189.1) node [anchor=north west][inner sep=0.75pt]  [font=\footnotesize]  {$I_{1}$};
% Text Node
\draw (195.5,241.1) node [anchor=north west][inner sep=0.75pt]  [font=\footnotesize]  {$L^{-}( \partial I_{1})$};

\end{tikzpicture}
\end{center}

It must be noted that the quantum expansion is satisfying the quantum focusing condition \cite{Bousso:2015mna}, namely that analogously to that of the classical expansion, under infinitesimal deformations of a surface along a null ray the quantum expansion cannot increase:
\begin{equation}
    \frac{\delta}{\delta V\left(y\right)} \Theta\left[V(y) ; y_{1}\right] \leq 0.
\end{equation}
Here, the positive function $V(y)$ defines the null hypersurface $L$ with $y\in I$. Then, we can define the quantum expansion $\Uptheta [V(y); y_{1}]$ as in equation \eqref{eq:qexp}, in terms of another point $y_{1}$ around which we are considering the "patch" area $\mathcal{A}$. Then, the \textit{quantum Bousso bound} \cite{Strominger:2003br, Bousso:2014sda} is the statement that along a light-sheet, the generalized entropy must be monotonically decreasing if the quantum expansion is negative. Finding a surface $\sigma '$, we require that the quantum expansion does not become positive, and that the generalized entropy is decreasing throughout the hypersurface. Then, defining the generalized entropy $S_{gen}(\sigma , \Sigma )$ as \cite{Bousso:2015mna}:
\begin{equation}
    S_{gen}(\sigma , \Sigma )=\frac{A(\sigma )}{4G\hbar }+S_{ext}(\sigma , \Sigma ),
\end{equation}
we can say that at $\sigma '$, we would have
\begin{equation}
    \Delta S_{ext}(\sigma , \sigma ')\leq \frac{\Delta A(\sigma , \sigma ')}{4G\hbar }.
\end{equation}
This would imply a \textit{quantum null energy condition}, which states that following holds:
\begin{equation}
    \langle T_{kk}\rangle \geq \lim _{\
    \mathcal{A}\to 0}\frac{\hbar }{2\pi \mathcal{A}}\frac{d^{2}S_{ext}}{d\lambda ^{2}}.
\end{equation}
It must be noted here that the quantum focusing condition in itself makes not much difference to the proof in the theorem \ref{qba} than the assumption that the generalized entropy must be decreasing and must terminate the null generator segments at some affine value. However, in understanding the effect of the generalized second law as a constraint on the topological conditions on $(M, g)$, this has a very significant role. For instance, one would expect the generalized entropy to be in terms of the QFC, which has been shown in the case of holographic screens (or \textit{marginally trapped tubes}) in \cite{Bousso:2015mqa, Bousso:2015qqa, Bousso:2015eda}.
We have not considered the future outgoing null congruence in a strict sense -- we have only made statements about the $l^{\mu }$ congruence. If we consider the outgoing congruence to have a vanishing expansion, we would refer to a marginally trapped surface. In this case, we would expect the spacetime to contain a horizon, under which we can construct a holographic screen (quantum holographic screen if the surfaces are quantum marginally trapped surfaces), which would equip the codimension $1$ surface formed by the foliation of these marginally trapped surfaces with an area law that implies a generalized second law. This is based off the QFC, which ensures the GSL is implied throughout the foliation of marginally trapped surfaces (called "leaves"). In a paper on the use of holographic screens in cosmological evolution \cite{Carroll:2017kjo}, it was shown that the holographic screens form of the area law predicts that the late evolved states of cosmologies can be found to be that of de Sitter spacetime, implying the cosmic no-hair theorem. Further, in the case of a quantum trapped surface, we can invoke Wall's theorem \cite{Wall:2010jtc} to find out null geodesic incompleteness in the spacetime. One can define Wall's version directly in terms of $\Uptheta $, which would have the same effect as stating that the fine-grained $S_{gen}(\sigma , \Sigma )$ is decreasing, but with the constraint that this variation is monotonic, which we can impose via the QFC. In fact, Wall's semiclassical generalization can be extended from the non-compact case of $I$ to a compact case under the assumption that along the $K$ congruence the expansion is also negative. Then, we have the following theorem \cite{Wall:2010jtc}:
\begin{theorem}
If a globally hyperbolic spacetime contains a surface $K$ as mentioned above and there is a non-compact Cauchy surface $H$, then there exists a null geodesic incompleteness.
\end{theorem}
\begin{proof}
The proof of this theorem is rather straightforward, with mostly the natural proof of the Penrose theorem proving that such a spacetime contains a geodesic incompleteness. Since the fine-grained generalized entropy is decreasing on the null hypersurface, the null generator segments must terminate at some affine value $\lambda $. The next parts of the proof are directly that of Penrose's, and the topological constraint is that a non-compact slice $\Sigma $ must not evolve into a compact $K$ due to global hyperbolicity. Since $I$ is a trapped surface both the future null orthogonal congruence hypersurfaces must be compact (one can define this in terms of the quantum expansion, as we will see in the quantum expansion discussion), and by defining a timelike vector field on $(M, g)$, one can see that the corresponding continuous $1-1$ mapping defined by the boundary of the causal domain $J(I)$ would have to be compact, but this would mean that there would be a boundary of the image $\varphi (\Dot{J}^{+}(I))$ in $H$, since $H$ is a non-compact Cauchy surface \cite{Hawking:1973uf}. Due to this, there must exist an incomplete null geodesic.
\end{proof}

Under the quantum Bousso bound, we have the condition that the variation of the exterior entropy $\Delta S_{ext}(\sigma , \sigma ')$ is always less than or equal to one-fourth of the variation of the interior entropy. If the light-sheet was allowed to terminate when $I$ is a hyperentropic region, we would have a violation of the Bousso bound. Naturally, the quantum focusing condition can be integrated to get an entropy bound on the light-sheet. Due to this, it must be necessary to condition singularities in cases of hyperentropic regions so as to prevent a violation of the Bousso bound. The null energy condition was implied to ensure that the classical expansion is non-positive throughout $L$ when $\theta _{0}(\partial I)\leq 0$, so that the only cases where the light-sheet can terminate are at caustics or singularities.

\textbf{Remarks:} In this paper, we have put forward a quantum version of the theorem introduced in \cite{Bousso:2022cun}. We have discussed and proved the classical and quantum versions of the theorem. We have shown that the result by Bousso and Moghaddam puts forward geodesic incompleteness as a response of spacetime to hyperentropic regions, and discussed the generalized entropy as an analogous tool in showing the compactness and the equivalence of the light-sheet entropy bound to the entropy intersecting the surface. If the light-sheet closes off, the curves intersecting the surface would also intersect $L$, and by a homeomorphic mapping one could say that the timelike curves would also intersect $L$, which would violate the Bousso bound. Furthermore, the role of the quantum focusing conjecture in singularity theorems is an interesting one, where the role of the classical expansion is replaced by a stronger condition on the generalized entropy of the surface. It will be of significance to see how the exterior entropy affects the overall geometry in such singularity theorems, and how the quantum focusing condition governs scenarios like non-locality. 

\bibliography{apssamp}% Produces the bibliography via BibTeX.

\end{document}